\documentclass[aps,twocolumn,superscriptaddress, nofootinbib]{revtex4-2}

\usepackage{amsmath}
\usepackage{amssymb}
\usepackage{amsthm}
\usepackage{color}
\usepackage{enumerate}
\usepackage{enumitem}
\usepackage{tabularx} 
\usepackage[pdftex]{graphicx}
\usepackage{braket}
\usepackage{ulem}

\def\II{\mathbb I}

\newtheorem{Lmm}{Lemma}
\newtheorem{Thm}{Theorem}
\newtheorem{Dfn}{Definition}
\newtheorem{Crl}{Corollary}

\newcommand{\expect}[1]{\left\langle #1 \right\rangle} 

\begin{document}

\title{
Tight scaling of key rate for differential-phase-shift quantum key distribution
}

\author{Akihiro Mizutani}
\affiliation{
Faculty of Engineering, University of Toyama, Gofuku 3190, Toyama 930-8555, Japan
}

\author{Toyohiro Tsurumaru}
\affiliation{Mitsubishi Electric Corporation, Information Technology R\&D Center,\\
5-1-1 Ofuna, Kamakura-shi, Kanagawa, 247-8501, Japan}

\begin{abstract}
The performance of quantum key distribution (QKD) protocols is evaluated based on the ease of implementation and 
key generation rate. Among major protocols, the differential-phase-shift (DPS) protocol has the advantage of 
simple implementation using a train of coherent pulses and a passive detection unit. Unfortunately, however,
 its key rate is known to be at least proportional to $\eta^2$ with respect to channel transmission $\eta\to0$. 
If one can only prove the rate proportional to $\eta^2$ and cannot improve the analysis beyond that, 
then the DPS protocol will be deemed inferior to other major protocols, such as the decoy BB84 protocol.
In this paper, we consider a type of DPS protocol in which the phase of each emitted block comprising $n$ pulses is randomized and significantly improve the analysis of its key rate. Specifically, we reveal that the 
key rate is proportional to $\eta^{1+\frac{1}{n-2}}$ and this rate is tight. 
This implies that the DPS protocol can achieve a key rate proportional to $\eta$ for a large number of $n$, which is the same scaling as the decoy BB84 protocol. 
Our result suggests that the DPS protocol can achieve a combination of both advantages of ease of implementation 
and a high key generation rate.
\end{abstract}

\maketitle

\section{Introduction}
\label{sec:intro}
Quantum key distribution (QKD) enables information-theoretically secure communication by distributing a secret key 
between Alice and Bob~\cite{Lo2014}. 
Among major QKD protocols, the differential-phase-shift (DPS) protocol~\cite{dps1,dps2} has 
an advantage of simple implementation involving a train of coherent pulses as well as the passive detection unit 
without dynamic switching of the basis for each received pulse. 
The simplicity of this protocol attracts intensive attentions from 
theory~\cite{curty,tsurumaru,wen2009,Endo_2022,wolf,tdps1,tdps2,tdps3,tamaki2012unconditional,Mizutani_2018} 
and experiments~\cite{ex1,ex2,ex3,tokyo} 
including field demonstration in the Tokyo QKD network~\cite{tokyo}. 
In the DPS protocol, Alice sends a train of coherent pulses and encodes bit information 
into the relative phase between the adjacent pulses. 
Such a train of coherent pulses can be generated by chopping continuously oscillating lasers into pulse shapes, which 
has a high affinity with current optical communication and an ability of high-speed transmission. 
Employing such coherent pulses is advantageous over other major protocols, such as the decoy 
BB84 protocol~\cite{decoy}, which employs phase-randomized coherent states generated by oscillating the laser for each pulse. 
To completely eliminate the phase coherence among the emitted pulses, the speed of transmission becomes slower than 
the one with a train of coherent pulses. 
In addition to the advantage of fast transmission, the DPS protocol has the advantage of being robust against source imperfections~\cite{tdps1,tdps2,tdps3}, and any independent and identical states can be securely employed.

While the DPS protocol has the advantage of fast and ease of implementation, 
the key generation rate under the information-theoretic security is known to be 
proportional to $\eta^2$ or higher in the limit $\eta\to0$ (namely, in the high-loss regime)~\cite{tdps1,tdps2,tdps3}.
Here, $\eta$ represents the transmission of the quantum channel between Alice and Bob.
This evaluation of the key rate is rather low, considering the fact that the theoretical limit of the key rate is 
of order $O(\eta)$ in the high-loss regime~\cite{Takeoka_2014,Pirandola2017}. 
The scaling of this theoretical limit is achieved by the decoy BB84 protocol~\cite{decoy}, and hence 
one may have to conclude that the DPS protocol is inferior to other major protocols in terms of the key rate. 
To claim the superiority, it is crucial to improve the evaluation of the scaling behavior, or the exponent of $\eta$. 
This importance 
can be seen, e.g., by noting that, if a protocol with the key rate being proportional to $\eta^2$ improves tenfold over a 
protocol proportional to $\eta$ at zero distance, this advantage disappears at 50km, and beyond that distance, 
the gap exponentially increases.

In this paper, we consider a type of DPS protocol studied in~\cite{tamaki2012unconditional,Mizutani_2018} 
where consecutive $n$ emitted pulses are regarded as a block and at most 
one bit of secret key is extracted from each block.
We establish the actual performance of this protocol by 
deriving the lower and upper bounds on the key rate and demonstrate that these bounds coincide 
when quantum bit error rate is zero. 
Specifically, we derive the upper bound by using the fact that 
the probability of Eve succeeding in finding all the relative phases between the adjacent emitted pulses must be less 
than Bob's detection rate to obtain a positive secret key rate.
Also, we derive the lower bound on the key rate by providing an information-theoretic security proof by reducing 
it to the proof with collective attacks using the de Finetti representation theorem~\cite{renner2006security}. 
As a result, we find that the key rate of the DPS protocol can achieve proportional to $\eta^{1+\frac{1}{n-2}}$ with 
$n\ge3$ and this rate is tight. 
Note that 
the examples of QKD protocols with matching the scalings of the lower and upper bounds are the decoy BB84 and coherent-one-way 
(COW) protocols. Specifically, the optimal key rate for the COW protocol is proportional to 
$\eta^2$~\cite{securityCOW,attackCOW}, 
and the optimal rate for the decoy BB84 protocol is proportional to $\eta$~\cite{decoy,Takeoka_2014,Pirandola2017}.
Our result means that the key rate of the DPS protocol is proportional to $\eta$ 
with a sufficiently large block size $n$, which is the same scaling as the decoy BB84 protocol~\cite{decoy} and 
the theoretical limit~\cite{Takeoka_2014,Pirandola2017} of the QKD protocol in the high-loss regime.

\begin{figure*}[t]
\centering
\includegraphics[width=12cm]{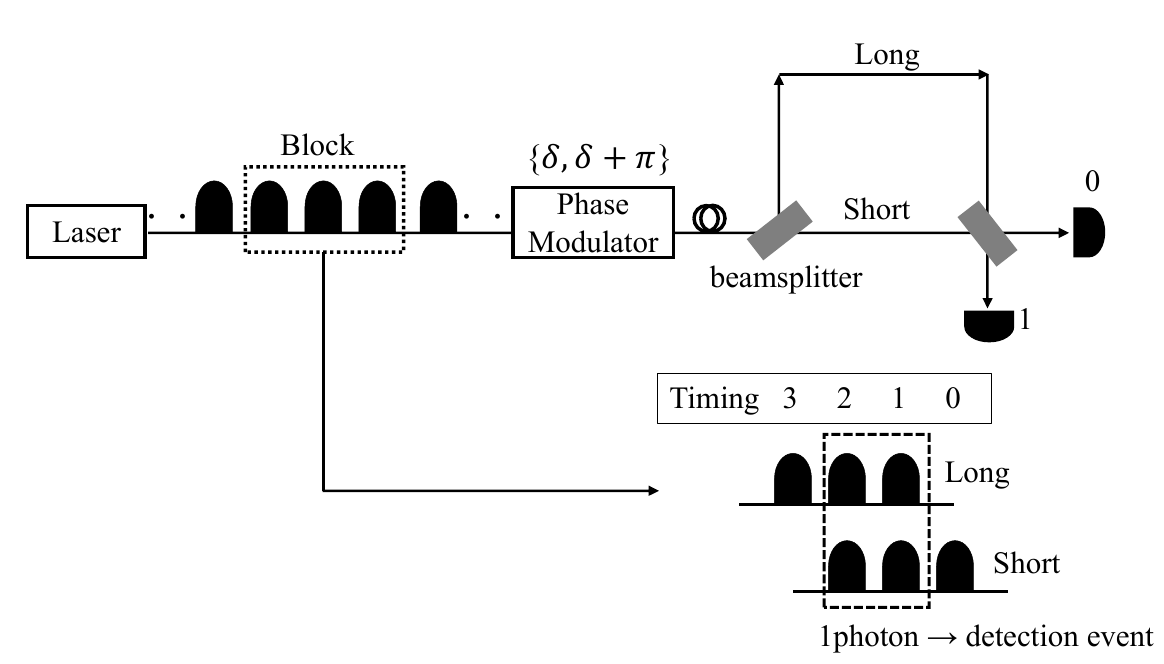}
\caption{
Setup of our DPS protocol with the block size $n=3$. The phase of each 
coherent pulse is modulated by a phase modulator that randomly modulates a phase $\delta$ or $\delta+\pi$ 
with $\delta$ being randomly chosen from $[0,2\pi)$ for each block. The emitted blocks are forwarded to a one-bit delay 
Mach-Zehnder interferometer with long and short paths and two 50:50 beamsplitters. 
The pulse trains leaving the interferometer are measured by two single 
photon detectors which report bit values 0 and 1. A detection event is defined as the event where Bob detects one photon 
in total among the timings 1, 2 and detects the vacuum in all the other timings.
}
 \label{fig:protocol}
\end{figure*}

\section{Protocol description}
\label{sec:setting}
Of various versions of the DPS protocols, 
we here particularly consider the one given in Ref.~~\cite{tamaki2012unconditional}. 
We denote by $n$ the block size with $n\ge3$, or the number of optical pulses per block. 
Our protocol is depicted in Fig.~\ref{fig:protocol}. 

\paragraph{States sent by Alice}
Alice generates random bits $\vec{s}:=s_1\dots s_n\in\{0,1\}^n$, random number $\delta\in[0,2\pi)$ 
and a block $\bigotimes_{i= 1}^n\ket{e^{\rm{i}\delta}(-1)^{s_i}\alpha}$ consisting of $n$ coherent states, 
and she sends the block to Bob.
As the uniform random phase $\delta$ is applied on the entire block here, for each block, the total photon number becomes a classical variable $\nu\in\{0,1,2,...\}$.
Thus, when viewed from the outside, Alice's operation above  is equivalent to emitting
\begin{eqnarray}
\ket{\psi_{\vec{s}\nu}}&:=&\frac1{\sqrt{\nu!}}\hat{a}_{\vec{s}}^{\dagger\nu}\ket{0},
\label{eq:psi}
\\
\hat{a}_{\vec{s}}&:=&\frac1{\sqrt{n}}\sum_{i=1,\dots,n}(-1)^{s_i}\hat{a}_i
\end{eqnarray}
with the probability
\begin{equation}
p_{\mu,\nu}=e^{-\mu}\frac{\mu^\nu}{\nu!},
\label{eq:prop}
\end{equation}
where $\ket{0}$ is the vacuum state and $\mu=n|\alpha|^2$. 
Here, $\hat{a}_i$ denotes the annihilation operator of the $i$th emitted state.

\paragraph{Bob's detection}
Bob then receives the block of pulses and measures it with single photon detectors at timings $i=0,1,\dots,n$.
If he detects one photon at one timing $i=1,\dots,n-1$ (i.e., except at both end timings) and no photon elsewhere, 
he considers it as a $i$th detection event and announces $i$; otherwise he announces ``no detection.''

\paragraph{Sifting}
When the $i$th detection event occurs, Bob defines his sifted key $b'$ depending on which of the two 
detectors clicks, and Alice defines her sifted key as
\begin{equation}
b=s_i\oplus s_{i+1}.
\end{equation}
If there is no detection, they define $b=b'=\bot$. 

\paragraph{Parameter estimation}
After repeating the steps above multiple times, based on the measurement results, Alice and Bob estimate pulse detection rate probability $r$ and the quantum bit error rate (QBER) $e$ in the detection events.

As customary in analyses of QKD protocols, throughout the paper we  assume that the quantum communication channel is amplitude damping.
In this case, $r$ should equal on average
\begin{equation}
r=\frac{n-1}{n}e^{-\eta\mu}\eta\mu
\label{eq:reaching_fraction}
\end{equation}
with $\eta$ being the transmission of the channel. 
That being said, however, this does not imply that our security proof assumes the quantum channel between Alice and Bob. 
If the estimated detection rate $r$ deviates from the expected value, Alice and Bob simply abort the protocol.

\paragraph{Key distillation}
Based on the result of the parameter estimation, Alice and Bob determine the specifications of the error correction code and the hash function for privacy amplification, and perform the key distillation to generate the secret key.

\section{main result}
In this section, we first introduce notations related to the upper and lower bounds on the key rate for our DPS 
protocol and state our main result of this paper. \\
\begin{Dfn}
\label{dfn:key_rate_exponent}
Let $G(n,e,\eta)$ be the supremum (lowest upper bound) of the key rate for block size $n$, QBER $e$, and the channel transmission $\eta$. 
If $G(n,e,\eta)$ is upper-bounded by $k\eta^c$ for some positive constants $k$ and $c$
in the high-loss regime $\eta\to0$, using the Landau notations we write
$$G(n,e,\eta)=O(\eta^c).$$
Similarly, if  $G(n,e,\eta)$ is lower-bounded by $k\eta^c$ for some positive constants $k$ and $c$
in the high-loss regime, we write
$$G(n,e,\eta)=\Omega(\eta^c).$$
\end{Dfn}
We give the rigorous definition of the key rate in Appendix~\ref{sec:def_key_rate_exponent}. 
With this notation, our main result can be stated as follows.

\begin{Thm}
\label{thm:main_theorem}
The key rate of the DPS protocol introduced in the Sec.~\ref{sec:setting} is upper-bounded as
\begin{align}
G(n,e,\eta)=O(\eta^{1+\frac{1}{n-2}}).
\label{eq:upper}
\end{align}
This bound is tight for $e=0$, as its lower bound is
\begin{align}
G(n,0,\eta)=\Omega(\eta^{1+\frac{1}{n-2}}).
\label{eq:lower}
\end{align}
\end{Thm}

Takeoka et al. \cite{Takeoka_2014} and Pirandola et al.~\cite{Pirandola2017} 
showed that the upper bound on the key rate of any QKD protocol is of order $O(\eta)$ in the high-loss regime, 
and the decoy BB84 protocol~\cite{decoy} was shown to achieve this optimal key rate. 
Our Theorem~\ref{thm:main_theorem} states that the DPS protocol can also achieve the optimal key rate 
for a sufficiently large block size $n$. 
Recalling that the DPS protocol is easier to implement than other major protocols (cf. the first paragraph of Sec.~\ref{sec:intro}), this result suggests that the DPS protocol is particularly promising and advantageous. 

Note that in the previous works of the DPS protocol with block-wise phase randomization~\cite{tamaki2012unconditional,Mizutani_2018}, 
they showed that the key rate is proportional to $\eta^{3/2}$ by proving that the secret key can be generated from the 
events where Alice emits two photons in one block. 
Our lower bound in Eq.~(\ref{eq:lower}) is consistent with these existing results when $n=3$, 
but our theorem implies that the key rate can be improved when we increase the block size $n$. 
This is specifically because not only from two and three photon emission events whose security is discussed in~\cite{tamaki2012unconditional,Mizutani_2018}, we could also generate the secret key from $\nu$-photon 
emission events with $\nu\le n-2$. 

The rest of this paper is devoted to proving Theorem~\ref{thm:main_theorem}. 
In Sec.~\ref{sec:upperbound}, we derive the upper bound on the key rate shown in Eq.~(\ref{eq:upper}) by 
explicitly constructing Eve's attack. Next, in Sec.~\ref{sec:lowerbound}, we derive the lower bound on the 
key rate shown in Eq.~(\ref{eq:lower}). 
The main technical lemma for deriving the lower bound is Lemma~\ref{lmm:HAE_lower_bound}, which we 
prove in Sec.~\ref{sec:lemma1}. 
Lemma~\ref{lmm:HAE_lower_bound} states that 
the entropy of Alice's sifted key bit viewed by Eve when Alice emits $\nu$ ($\nu\le n-2$) photons in a block is positive, 
namely, when Eve has uncertainly about Alice's key, we could generate the secret key. 

\section{Upper bound on the key rate}
\label{sec:upperbound}
We begin by proving the first half of Theorem \ref{thm:main_theorem}.
That is, we prove Eq.~(\ref{eq:upper}) for any QBER $e\ge0$.

Consider the following attack:
Eve measures the state sent out by Alice, using the same measurement device as Bob's.
If Eve succeeds in detecting a photon in all timings $i=1,2,\dots,n-1$ and thereby finding all values $s_i\oplus s_{i+1}$, 
then Eve 
generates the correct single photon state $\ket{\psi_{\vec{s}1}}$ and sends it to Bob; otherwise sends the vacuum state to Bob. If the condition
\begin{itemize}
\item[i)] Success probability $P_{\rm Eve}$ of Eve's measurement\\
$\ge$ Bob's detection rate $r$ in Eve's absence\\ 
\quad [given in Eq.~(\ref{eq:reaching_fraction})]
\end{itemize}
is satisfied in this attack, then Eve successfully finds all information related with the sifted key, $s_i\oplus s_{i+1}$, without being noticed by Bob.
In other words, Eve succeeds in the perfect attack, and thus the key rate should vanish, $G(n,e,\eta)=0$. 
By taking its contraposition, the necessary condition for $G(n,e,\eta)>0$ is $P_{\rm Eve}<r$, which leads to 
\begin{align}
P_{\rm Eve}<r=\frac{n-1}{n}e^{-\eta\mu}\eta\mu\le\eta\mu.
\label{eq:ever}
\end{align}
If Alice emits the $n-1$-photon state $\ket{\psi_{\vec{s}n-1}}$, 
Eve has a non-zero chance of detecting a photon for all timings $i=1,2,\dots,n-1$ when using the same device as Bob's. 
Hence, we have $p_{\mu,n-1}q_{n-1}\le P_{\rm Eve}$ with $q_{n-1}>0$ 
denoting the probability that Eve succeeds in finding all $s_i\oplus s_{i+1}$ given $\ket{\psi_{\vec{s}n-1}}$. 
Applying this inequality to Eq.~(\ref{eq:ever}) along with Eq.~(\ref{eq:prop}) yields 
$e^{-\mu}\frac{\mu^{n-1}}{(n-1)!}q_{n-1}<\eta\mu$. 
Using the fact that $\mu$ is upper-bounded by some positive constant in the high-loss regime
\footnote{
We prove this by contradiction. 
If this does not hold, it means that even if $\eta\to 0$, $\mu$ is greater than some positive constant. 
This clearly contradicts Eq.~(\ref{eq:ever}) because $\eta\mu$ approaches 0, 
but $P_{\rm Eve}$ is still a certain positive constant.
}, we obtain the upper bound on Alice's pulse intensity $\mu$ using the Landau symbol as
\begin{align}
\mu=O(\eta^{\frac{1}{n-2}}).
\end{align}
Combining this relation and the fact that the key rate $G(n,e,\eta)$ 
should be less than Bob's detection rate $r$ ($\le\eta\mu$), we 
arrive at Eq.~(\ref{eq:upper}). This ends the proof of the first half of Theorem~\ref{thm:main_theorem}.

\section{lower bound on the key rate}
\label{sec:lowerbound}
Next, we prove the second half of Theorem \ref{thm:main_theorem}, namely, Eq.~(\ref{eq:lower}).

\subsection{Reduction to collective attacks}
We derive a lower bound on the key rate for the asymptotic total number of blocks $N\to \infty$. 
In so doing, if we could apply the de Finetti representation theorem~\cite{renner2006security} directly to our protocol, 
the analyses would become simple. 
This is because the theorem guarantees that an asymptotic key rate evaluated against collective attacks would also serve 
as that against coherent attacks (the most general attack). 

However, the de Finetti representation theorem cannot be directly applied to our protocol, since in our DPS protocol Alice sends out infinite dimensional states $\ket{\psi_{\vec{s}\nu}}$ for $\nu\ge0$.
To circumvent this problem, we consider the following variant of our protocol:
Recall that, as mentioned in Sec.~\ref{sec:setting}, 
Alice's operation involving Eq.~(\ref{eq:psi}) is equivalent to first choosing the classical random variable $\nu\in\{0,1,2,\dots\}$ with probability $p_{\mu,\nu}$, and then sending out a $\nu$-photon state $\ket{\psi_{\vec{s}\nu}}$.
We modify the operation there for $\nu\ge n-1$ as follows.
\begin{description}
\item[At-most-$n-2$-photon protocol]
If the state emitted by Alice's device contains 
at most 
$n-2$ photons (i.e., $\ket{\psi_{\vec{s}\nu}}$ with $\nu\le n-2$), then she sends it out as is.
Otherwise (i.e., if the state contains more than $n-1$ photons), she replaces it with the classical data 
$\vec{s}\in\{0,1\}^n$ and sends it out.
\end{description}
Below, we will often refer to this protocol simply as the modified protocol.
As we elaborate below,
\begin{enumerate}[label=Item \arabic*:]
\item this modified protocol allows the use of the de Finetti representation theorem, and
\item the key rate thus obtained serves as a lower bound on the asymptotic key rate for the original protocol.
\end{enumerate}

Item 1 is true because in this modified protocol, Alice always sends out finite-dimensional states, consisting of $n-2$ photons states, or $2^n$-dimensional state for sending $\vec{s}$.

More precisely, the DPS protocol under consideration fits within the framework of Sec.~6.5.1, Ref. \cite{renner2006security}, if we (i) regard one block in our modified protocol, consisting of $n$ pulses, as one big pulse, and (ii) rewrite the protocol, which has the prepare-and-measure structure, as an entanglement-based protocol with a fixed marginal on Alice's side (for details, see Ref.~\cite{George_2021}, Sec.~III. A). 
As a result, we can apply Corollary 6.5.2 of Ref. \cite{renner2006security} to our modified protocol, and show that its key rates against collective attacks and against coherent attacks match in the limit $N\to\infty$.

Item 2 above can be seen as follows.
Note that when Alice's device emits no less than $n-1$ photons, Eve is informed of all information $\vec{s}$ related with the sifted key.
Eve may use this $\vec{s}$ to generate the state that Alice was supposed to send in the original protocol.
This means that, in the modified protocol, Eve maintains the freedom to mount any attacks that are possible in the original protocol.
In other words, Eve is no less strong in the modified protocol than in the original protocol; thus the modification of the protocol does not increase the key rate.
Therefore, the lower bound on the key rate for the modified protocol also serves as that for the original protocol.

\subsection{Analysis against collective attacks}
\label{subsec:ana}
Thus we analyze collective attacks to our modified protocol.
The collective attacks here refer to those in which Eve performs the same quantum operation on each block, consisting of $n$ pulses, rather than on each pulse. 
In the following, we describe the modified protocol under which Eve performs the collective attacks 
and formulate the states and operations of Alice, Bob, and Eve at each step.

\subsubsection{Formulation}
\label{sec:formu}
\paragraph{States sent by Alice}

As explained in the previous subsection, in our modified protocol, Alice chooses a classical variable $\nu\in\{0,1,2,\dots, n-1\}$ randomly, and then sends out a $\nu$-photon state $\ket{\psi_{\vec{s}\nu}}$ if $\nu\le n-2$, and $\vec{s}\in\{0,1\}^n$ 
otherwise.
Since states $\ket{\psi_{\vec{s}\nu}}$ and $\vec{s}$ corresponding to different values of $\nu$ live in different Hilbert spaces, Eve can measure $\nu$ without affecting the states.
Therefore, without loss of security (i.e., without increasing the key rate), we may assume that $\nu$ was in fact announced to Eve by Alice in advance.
Hence the initial state prepared by Alice can be described as
\begin{eqnarray}
\hat{\rho}^{\rm ini}&=&\sum_{\nu=0}^{n-1}p'_{\mu,\nu}\hat{\rho}^{\rm ini}_\nu\otimes\ket{\nu}\bra{\nu}_{E_1},
\label{eq:rho_ini_def}
\end{eqnarray}
with
\begin{eqnarray}
\hat{\rho}^{\rm ini}_\nu&:=&\frac1{2^{n}}\sum_{\vec{s}\in\{0,1\}^n}\ket{\vec{s}}\bra{\vec{s}}_{A}\nonumber\\
&&\quad\otimes
\left\{
\begin{array}{lll}
\ket{\psi_{\vec{s}\nu}}\bra{\psi_{\vec{s}\nu}}_{B}&{\rm for}&0\le \nu\le n-2,\\
\ket{\vec{s}}\bra{\vec{s}}_B&{\rm for}&\nu=n-1,
\end{array}
\right.\\
p'_{\mu,\nu}&:=&
\left\{
\begin{array}{lll}
p_{\mu,\nu}&{\rm for}&0\le \nu\le n-2,\\
\sum_{\nu=n-1}^\infty p_{\mu,\nu}&{\rm for}&\nu=n-1.
\end{array}
\right.
\end{eqnarray}
Here, system $E_1$, storing the information of $\nu$, is possessed by Eve.

\paragraph{Eve's eavesdropping operation}
Eve's eavesdropping operations on the quantum channel can generally be described as quantum operations 
$\hat{M}^{\rm eav}:BE_1\to BE_1E_2$ on state $\hat{\rho}^{\rm ini}$ defined above.
We note here that, since the input $\hat{\rho}^{\rm ini}$ is block diagonalized as 
shown in Eq.~(\ref{eq:rho_ini_def}), we do not lose generality by assuming that, on each $\hat{\rho}^{\rm ini}_\nu$ having the 
photon number $\nu$, Eve individually applies an operation $\hat{M}^{\rm eav}_\nu:B\to BE_2$.

\paragraph{Sifting}
Subsequently, Alice and Bob respectively extract the sifted keys $b,b'\in\{0,1,\bot\}$ as well as the data regarding the presence/absence of detection, $i\in\{1,\dots,n-1,\bot\}$, 
by using the procedure described in Sec.~\ref{sec:setting}.
They then also announce $i$.
We denote these entire procedures as the quantum operation $\hat{M}^{\rm sif}:AB\to ABE_3$.
Hence the resulting state $\hat{\rho}$ can be written as 
\begin{eqnarray}
\hat{\rho}&=&\hat{M}^{\rm sif}\circ \hat{M}^{\rm eav}(\hat{\rho}^{\rm ini})\nonumber\\
&=&\sum_{\nu=0}^{n-1}p'_{\mu,\nu}(\hat{\rho}_{\nu})_{ABE_2E_3}\otimes\ket{\nu}\bra{\nu}_{E_1}
\end{eqnarray}
with
\begin{align}
\label{eq:rhonu}
(\hat{\rho}_{\nu})_{ABE_2E_3}=&\hat{M}^{\rm sif}\circ \hat{M}^{\rm eav}_\nu(\hat{\rho}^{\rm ini}_\nu)\nonumber\\
=&\sum_{b,b'=0,1,\bot}\ket{b}\bra{b}_A\otimes\ket{b'}\bra{b'}_B\nonumber\\
&\quad\otimes \sum_{i=1,\dots,n-1,\bot}(\hat{\rho}_{\nu bb' i})_{E_2}\otimes\ket{i}\bra{i}_{E_3}.
\end{align}
Here, Eve's system $E_3$ stores the information of $i\in\{1,\dots,n-1,\bot\}$ that is announced by Bob. 
Also, $(\hat{\rho}_{\nu bb' i})_{E_2}$ represents the state of Eve's remaining system (except for the systems 
storing the information of $\nu$ and $i$) when Alice emits $\nu$ photons, the result of Bob's detection is 
$i$, and Alice and Bob respectively obtain $b$ and $b'$.

\paragraph{Parameter estimation}
Alice and Bob evaluate the pulse arrival rate $P^{\rm det}$ and the error rate (pulse arrival rate $\times$ error rate) 
$P^{\rm err}$ of the state $\hat{\rho}$ after sifting, which are defined as
\begin{eqnarray}
P^{\rm det}&=&{\rm Tr}(\hat{\rho}^{\rm det}),\\
P^{\rm err}&=&{\rm Tr}(\hat{\rho}^{\rm err}).
\end{eqnarray}
Here, 
\begin{eqnarray}
\hat{\rho}^{\rm det}:=\sum_{i=1}^{n-1}\left(\ket{i}\bra{i}\right)_{E_3}\hat{\rho}
\end{eqnarray}
and
\begin{eqnarray}
\hat{\rho}^{\rm err}&:=&\sum_{\substack{b,b'=0,1,\\ b\ne b'}}{\rm Tr}_{AB}\left[\left(\ket{b}\bra{b}\right)_A\otimes\left(\ket{b'}\bra{b'}\right)_B\hat{\rho}\right]
\end{eqnarray}
are the unnormalized states corresponding to the pulse detection announcement (the state filtered under the condition that Bob announces pulse detection) and the unnormalized state corresponding to the bit errors, respectively.

\paragraph{Key distillation}
\label{sec:collective_key_distillation}
If $P^{\rm det}$ or 
$P^{\rm err}$ deviate from the natural values (i.e., those in Eve's absence) by more than a certain extent, Alice and Bob abort the protocol.
That is, they abort the protocol unless
\begin{eqnarray}
P^{\rm det}&\ge &f^{\rm det}r,
\label{eq:detection_rate}
\\
P^{\rm err}&\le &f^{\rm err}er,
\label{eq:error_rate}
\end{eqnarray}
where $r,e$ are the expected values of $P^{\rm det}$, $P^{\rm err}$ in the absence of Eve. 
For example, if the channel in Eve's absence is amplitude damping, $r$ is given by Eq.~(\ref{eq:reaching_fraction}).
The constants $f^{\rm det}$, $f^{\rm err}$ specify the degree to which the measured values of $P^{\rm det}$, $P^{\rm err}$ can deviate from $r, re$; e.g., if one tolerates deviation up to twice the natural values, $f^{\rm det}=1/2$, $f^{\rm err}=2$. 
Note that the security of the protocol, with determination of $f^{\rm det}$ and $f^{\rm err}$ according to the observed values 
of $P^{\rm det}$ and $P^{\rm err}$, is discussed in Appendix~\ref{sec:appVL}.

If they decide to continue the protocol, Alice and Bob perform key distillation on the obtained sifted key and generate the secret key.

\subsubsection{Key rate formula and some useful properties}
According to Corollary 6.5.2 of Ref. \cite{renner2006security}, the asymptotic key rate of our modified protocol 
formulated in Sec.~\ref{sec:formu} can be lower-bounded as
\begin{equation}
G(n,e,\eta)\ge \min_{\substack{\hat{M}^{\rm eav}\,{\rm subject\ to}\\ {\rm Eqs.}\ (\ref{eq:detection_rate}), (\ref{eq:error_rate})}}
\left(H(A|E)_{\hat{\rho}}-H(A|B)_{\hat{\rho}}\right),
\label{eq:collective_key_rate}
\end{equation}
where Eve's system $E$ consists of three parts: $E=E_1E_2E_3$. 

Here, given a quantum state $\hat{\rho}$ of systems $ABE$, the conditional quantum entropy $H(A|E)_{\hat{\rho}}$ 
is defined by
\begin{align}
\label{eq:defCE}
H(A|E)_{\hat{\rho}}:=H(AE)_{\hat{\rho}}-H(E)_{\hat{\rho}}
\end{align}
with
\begin{align}
H(AE)_{\hat{\rho}}:=H({\rm Tr}_B\hat{\rho}),~H(E)_{\hat{\rho}}:=H({\rm Tr}_{AB}\hat{\rho}),
\end{align}
where $H(\cdot)$ is the von Neumann entropy. Also, $H(A|B)_{\hat{\rho}}$ is defined similarly.
$H(A|E)_{\hat{\rho}}$ appearing on the right hand side of Eq.~(\ref{eq:collective_key_rate}) can further be bounded as
\begin{eqnarray}
H(A|E)_{\hat{\rho}}&=&{\rm Tr}(\hat{\rho}^{\rm det})H(A|E)_{\bar{\hat{\rho}}^{\rm det}}\nonumber\\
&=&\sum_{\nu=0}^{n-1} p'_{\mu,\nu}{\rm Tr}(\hat{\rho}_{\nu}^{\rm det})H(A|E)_{\bar{\hat{\rho}}_{\nu}^{\rm det}}
\nonumber\\
&\ge&\sum_{\nu=0}^{n-2} p'_{\mu,\nu} {\rm Tr}(\hat{\rho}_{\nu}^{\rm det})H(A|E)_{\bar{\hat{\rho}}_{\nu}^{\rm det}}.
\label{eq:lower_bound_HAE}
\end{eqnarray}
Here, $\hat{\rho}_{\nu}^{\rm det}$ denotes the unnormalized state corresponding to Alice emitting $\nu$ photons and Bob announces a detection, which is defined as
\begin{equation}
p'_{\mu,\nu}\hat{\rho}_{\nu}^{\rm det}
:={\rm Tr}_{E_1}\left[\left(\ket{\nu}\bra{\nu}\right)_{E_1}\hat{\rho}^{\rm det}\right].
\label{eq:p'}
\end{equation}
Hereafter, the overline on a state means that it is normalized, e.g., 
$\bar{\hat{\rho}}^{\rm det}=\left({\rm Tr}\left(\hat{\rho}^{\rm det}\right)\right)^{-1}
\hat{\rho}^{\rm det}$.

The equality of the first line of Eq.~(\ref{eq:lower_bound_HAE}) follows from the fact that, when Bob does not announce a detection, the corresponding sifted key, $b=\bot$, is known to Eve.
The equality of the second line holds because the photon number $\nu$ is also known to Eve.
The inequality in the third line is true because $\bar{\hat{\rho}}_{\nu}^{\rm det}$ is a classical-quantum (cq) state, 
and satisfies $H(A|E)_{\bar{\hat{\rho}}_{\nu}^{\rm det}}\ge0$.

In order to find a concrete lower bound on the key rate $G(n,e,\eta)$, we need to evaluate the right-hand side of Eq.~(\ref{eq:collective_key_rate}). 
In the statement of the second half of Theorem~\ref{thm:main_theorem}, as we assume QBER $e=0$, 
the second term of Eq.~(\ref{eq:collective_key_rate}) vanishes, $H(A|B)_{\hat{\rho}}=0$, and the first term 
$H(A|E)_{\hat{\rho}}$ serves as the lower bound on $G(n,e,\eta)$.
This term can be further bounded by the third line of Eq.~(\ref{eq:lower_bound_HAE}).
We note that the quantity $H(A|E)_{\bar{\hat{\rho}}_{\nu}^{\rm det}}$ appearing there satisfies the following.
\begin{Lmm}
\label{lmm:HAE_lower_bound}
For each $n\ge3$, there exists a constant $H_n>0$ for which the following inequality holds:
For any $\nu\le n-2$,
\begin{equation}
H(A|E)_{\bar{\hat{\rho}}_{\nu}^{\rm det}}\ge H_n,
\label{eq:HAE_nu_bounded_by_H}
\end{equation}
under any attack $\hat{M}^{\rm eav}$ by Eve satisfying $P^{\rm err}=0$ (i.e.,QBER $e=0$) on 
Alice's initial state $\hat{\rho}_\nu^{\rm ini}$.
\end{Lmm}

We will prove this lemma in the next Sec.~\ref{sec:lemma1}. 
This lemma implies that if Bob frequently receives states from Alice containing at most $n-2$ photons, 
then the key rate becomes positive. 
To realize such a situation, it is sufficient that the probability of Alice emitting 
more than $n-2$ photon is smaller than $P^{\rm det}$, with which Bob announces detection.
This observation leads to the following corollary, or the second half of Theorem \ref{thm:main_theorem}.

\begin{Crl}
\label{crl:upper_bound_on_c}
For QBER $e=0$, we have
\begin{equation}
G(n,0,\eta)\ge \left(f^{\rm det}r-\frac{\mu^{n-1}}{(n-1)!}\right)H_n.
\label{eq:lmm_HAE_bounded_by_1-e}
\end{equation}
By setting $\mu=\eta^{1/(n-2)}$ and $f^{\rm det}>\frac{n}{n-1}\frac1{(n-1)!}$, we have
\begin{equation}
G(n,0,\eta)=\Omega(\eta^{1+\frac{1}{n-2}}).
\label{eq:coroG}
\end{equation}
\end{Crl}

\begin{proof}
From Eqs.~(\ref{eq:collective_key_rate}), (\ref{eq:lower_bound_HAE}), and (\ref{eq:HAE_nu_bounded_by_H}), 
and the relation $H(A|B)_{\hat{\rho}}=0$, it follows that
\begin{eqnarray}
G(n,0,\eta)&\ge& \sum_{\nu=0}^{n-2}p'_{\mu,\nu}{\rm Tr}\left(\hat{\rho}_{\nu}^{\rm det}\right) H_n.
\end{eqnarray}
This lower bound can be rewritten by using Eq.~(\ref{eq:p'}) as
\begin{align}
&\left[\sum_{\nu=0}^{n-1}{\rm Tr}\left(\ket{\nu}\bra{\nu}_{E_1}\hat{\rho}^{\rm det}\right)-
p_{\mu,n-1}'{\rm Tr}\left(\hat{\rho}_{n-1}^{\rm det}\right)\right]H_n\notag\\
=&\left[{\rm Tr}\left(\hat{\rho}^{\rm det}\right)-p_{\mu,n-1}'{\rm Tr}\left(\hat{\rho}_{n-1}^{\rm det}\right)\right]H_n,
\end{align}
and applying Eq.~(\ref{eq:detection_rate}) gives its lower bound as

\begin{align}
\left[f^{\rm det}r-p'_{\mu,n-1}{\rm Tr}\left(\hat{\rho}_{n-1}^{\rm det}\right)\right]H_n
\ge\left(f^{\rm det}r-p'_{\mu,n-1}\right)H_n.
\end{align}
Then by using the inequality 
$p'_{\mu,n-1}=\sum_{\nu=n-1}^\infty p_{\mu,\nu}\le \frac{\mu^{n-1}}{(n-1)!}$
\footnote{
This can be proven as follows. 
As 
\begin{align}
\binom{n+m-1}{n-1}\ge1\Leftrightarrow
n(n+1)(n+2)\cdots (n+m-1)\ge m!,
\end{align}
we have 
\begin{align}
n(n+1)(n+2)\cdots (n+m-1)\ge m!
\end{align}
for $m\ge1$. Using this leads to
\begin{align}
\sum_{\nu=n-1}^\infty p_{\mu,\nu}&=
e^{-\mu}\frac{\mu^{n-1}}{(n-1)!}\left(1+\frac{\mu}{n}+\frac{\mu^2}{n(n+1)}+\cdots\right)\\
&\le e^{-\mu}\frac{\mu^{n-1}}{(n-1)!}\left(1+\frac{\mu}{1!}+\frac{\mu^2}{2!}+\cdots\right)\\
&=\frac{\mu^{n-1}}{(n-1)!}.
\end{align}
}, we obtain Eq.~(\ref{eq:lmm_HAE_bounded_by_1-e}).

Next, we prove Eq.~(\ref{eq:coroG}). 
By setting $f^{\rm det}>\frac{n}{n-1}\frac1{(n-1)!}$, we can write 
$f^{\rm det}=\frac{n}{n-1}\frac1{(n-1)!}+\beta$ with $\beta>0$. 
Substituting this to Eq.~(\ref{eq:lmm_HAE_bounded_by_1-e}) leads to
\begin{align}
G(n,0,\eta)\ge\frac{\mu}{(n-1)!}\left[\eta e^{-\eta\mu}\left(1+\gamma\right)-\mu^{n-2}\right]H_n
\label{eq1:coro}
\end{align}
with $\gamma:=\beta(n-1)(n-1)!/n>0$. 
If we set $\mu=\eta^{1/(n-2)}$, we have that the right hand side of Eq.~(\ref{eq1:coro}) equals
\begin{align}
\frac{H_n}{(n-1)!}\left[e^{-\eta^{\frac{n-1}{n-2}}}\left(1+\gamma\right)-1\right]\eta^{\frac{n-1}{n-2}}.
\label{eq2:coro}
\end{align}
Since $e^x\ge 1+x$ holds for $x\in\mathbb{R}$, we have the lower bound on Eq.~(\ref{eq2:coro}) as
\begin{align}
\frac{H_n}{(n-1)!}\left[\gamma-\eta^{\frac{n-1}{n-2}}(1+\gamma)\right]\eta^{\frac{n-1}{n-2}}.
\label{eq3:coro}
\end{align}
We take some positive constant $\zeta>0$ such that
\begin{align}
\eta_0:=\left(\frac{\gamma}{1+\gamma}\right)^{\frac{n-2}{n-1}}-\zeta>0
\end{align}
holds and define $k$ as
\begin{align}
k&:=\frac{H_n}{(n-1)!}\left[\gamma-\eta_0^{\frac{n-1}{n-2}}(1+\gamma)\right]\notag\\
&>\frac{H_n}{(n-1)!}\left[\gamma-\left(\frac{\gamma}{1+\gamma}\right)(1+\gamma)\right]\notag\\
&=0.
\end{align}
Then, for any $\eta$ of $\eta<\eta_0$, we obtain
\begin{align}
G(n,0,\eta)&\ge
\frac{H_n}{(n-1)!}\left[\gamma-\eta^{\frac{n-1}{n-2}}(1+\gamma)\right]\eta^{\frac{n-1}{n-2}}\notag\\
&\ge\frac{H_n}{(n-1)!}\left[\gamma-\eta_0^{\frac{n-1}{n-2}}(1+\gamma)\right]\eta^{\frac{n-1}{n-2}}\notag\\
&=k\eta^{\frac{n-1}{n-2}}.
\end{align}
Using the symbol $\Omega$ defined in Eq.~(\ref{eq:Omega}), this inequality implies Eq.~(\ref{eq:coroG}). 
This completes the proof of  Corollary~\ref{crl:upper_bound_on_c}.
\end{proof}

\section{Proof of Lemma \ref{lmm:HAE_lower_bound}}
\label{sec:lemma1}
Here, we prove Lemma~\ref{lmm:HAE_lower_bound}, our main technical lemma to derive the lower bound on the 
key rate given in Eq.~(\ref{eq:coroG}). 
We initiate the discussion by recalling that any quantum operation can generally be written as adding an ancilla 
to the initial state and then applying a unitary transformation.
Thus any attack $\hat{M}_\nu^{\rm eav}$ by Eve on the initial $\nu$-photon state $\ket{\psi_{\vec{s}\nu}}_{B}$ 
can be described as applying an arbitrary unitary transformation $\hat{U}_{\nu}$ on 
$\ket{\psi_{\vec{s}\nu}}_{B}\otimes\ket{0}_{E_2}$.

As mentioned in the last paragraph of Sec.~\ref{sec:setting}, the resulting state 
$\hat{U}_{\nu}(\ket{\psi_{\vec{s}\nu}}_{B}\otimes\ket{0}_{E_2})$ will then be measured by Bob. 
Alice and Bob will generate a sifted key only when Bob measures exactly one photon in total. 
As we also restrict ourselves with the case of $e=0$, Bob should detect no errors when a detection event occurs. 
Hence the amplitude in $\hat{U}_{\nu}(\ket{\psi_{\vec{s}\nu}}_{B}\otimes\ket{0}_{E_2})$ corresponding to the sifted key generation must contain strictly one photon in total and cause no error, and thus equal $\ket{\psi_{\vec{s}1}}$. 
Therefore, we have for any $\vec{s}=s_1\dots s_n\in\{0,1\}^n$,
\begin{eqnarray}
\lefteqn{\hat{U}_{\nu}\left(\ket{\psi_{\vec{s}\nu}}_{B}\otimes\ket{0}_{E_2}\right)}\nonumber\\
&=&\ket{\psi_{\vec{s}1}}_B\otimes\ket{\varphi_{\vec{s}\nu}}_{E_2}+\sum_{j}\ket{v_j}_B
\otimes \ket{\chi_{\vec{s}\nu j}}_{E_2},
\label{eq:U_nu_def}
\end{eqnarray}
where $\ket{v_j}_B$ are zero- or multi-photon states, and $\ket{\varphi_{\vec{s}\nu}}$, $\ket{\chi_{\vec{s}\nu j}}$ are not necessarily normalized or orthogonalized.
After Bob measures this state and Alice generates the sifted key $b\in\{0,1\}$, 
the mixed state 
\begin{align}
\hat{\rho}_{\nu}:={\rm Tr}_B(\hat{\rho}_{\nu})_{ABE_2E_3},
\end{align}
with $(\hat{\rho}_{\nu})_{ABE_2E_3}$ 
defined in Eq.~(\ref{eq:rhonu}), of Alice's sifted key and Eve's system becomes
\begin{eqnarray}
\hat{\rho}_{\nu}&=&(\hat{\rho}_{\nu}^{\rm det})_{AE_2E_3}+(\hat{\rho}_{\nu}^{\bot})_{AE_2E_3},
\label{eq:rho_AE}
\end{eqnarray}
with
\begin{align}
(\hat{\rho}_{\nu}^{\rm det})_{AE_2E_3}&=
\sum_{1\le i\le n-1}\left(\hat{\rho}_{\nu i}^{\rm det}\right)_{AE_2}\otimes\ket{i}\bra{i}_{E_3},
\label{eq:rho_nu_det}\\
\left(\hat{\rho}_{\nu i}^{\rm det}\right)_{AE_2}&:=\sum_{b=0,1}\ket{b}\bra{b}_A\otimes(\hat{\rho}_{\nu bi}^{\rm det})_{E_2},
\label{eq:rho_nui_det}\\
\left(\hat{\rho}_{\nu bi}^{\rm det}\right)_{E_2}
&:=
\frac1{n}\frac1{2^n}\sum_{\substack{\vec{s}\ {\rm satisfying}\\{s_i\oplus s_{i+1}}=b}}
\ket{\varphi_{\vec{s}\nu}}\bra{\varphi_{\vec{s}\nu}}_{E_2},
\label{eq:rho_nubi_det}\\
(\hat{\rho}_{\nu}^\bot)_{AE_2E_3}&:=\ket{\bot}\bra{\bot}_A\otimes
\left(\hat{\rho}_{\nu}-\hat{\rho}_{\nu}^{\rm det}\right)_{E_{2}}\otimes\ket{\bot}\bra{\bot}_{E_3}
.
\label{eq:rho_nu_bot}
\end{align}
Here and henceforth, we omit systems in the subscript of the operators if there is no confusion. 
Here, $\hat{\rho}_{\nu}^{\rm det}$ ($\hat{\rho}_{\nu i}^{\rm det}$) is an unnormalized state when Bob obtains a detection event 
(an $i$th detection event). Also, $\hat{\rho}_{\nu bi}^{\rm det}$ represents an unnormalized state when Bob obtains an 
$i$th detection event and Alice's sifted key is $b$. Note that $1/2^n$ in Eq.~(\ref{eq:rho_nubi_det}) is the 
probability of Alice selecting $\vec{s}\in\{0,1\}^n$, 
and $1/n$ there is the probability of obtaining the $i$th detection event conditioned that 
Alice emits $\nu$ photons and the received block by Bob contains one photon. 
In Eq.~(\ref{eq:rho_nu_bot}), $\hat{\rho}_{\nu}^\bot$ is the unnormalized state for no detection. 
With these definitions, we see that ${\rm Tr}\hat{\rho}_{\nu}^{\rm det}$ is equal to the probability of 
obtaining a detection event conditional on Alice emitting $\nu$ photons.

With Eqs.~(\ref{eq:rho_AE})-(\ref{eq:rho_nu_bot}), the corresponding conditional entropy can be decomposed as
\begin{eqnarray}
\lefteqn{H(A|E)_{\hat{\rho}_{\nu}}=}\nonumber\\
&&{\rm Tr}(\hat{\rho}_{\nu}^{\rm det}) 
H(A|E)_{\bar{\hat{\rho}}_{\nu}^{\rm det}}+{\rm Tr}(\hat{\rho}_{\nu}^{\bot}) H(A|E)_{\bar{\hat{\rho}}_{\nu}^{\bot}},
\end{eqnarray}
and the second term on the right equals zero, $H(A|E)_{\bar{\hat{\rho}}_{\nu}^{\bot}}=0$, as can be seen from 
Eq.~(\ref{eq:rho_nu_bot}).

Thus it suffices to show that $H(A|E)_{\bar{\hat{\rho}}_{\nu}^{\rm det}}$ has a positive lower bound.
However, rather than tackling this problem directly, we will further simplify the problem. 
To this end, we first prove the following inequality:
\begin{equation}
\dim {\rm Span}\{\ket{\varphi_{\vec{s}\nu}}\}_{\vec{s}}\le \dim {\rm Span}\{\ket{\psi_{\vec{s}\nu}}\}_{\vec{s}}.
\label{eq:dim_span_varphi}
\end{equation}
Here, ${\rm Span}\{\ket{\varphi_{\vec{s}\nu}}\}_{\vec{s}}$ denotes the vector space spanned by 
$\{\ket{\varphi_{\vec{s}\nu}}\}_{\vec{s}}$, and dim$F$ denotes the dimension of vector space $F$. 
Equation~(\ref{eq:dim_span_varphi}) can be proven as follows:
Due to the general property of vector spaces, a certain subset of $\{\ket{\psi_{\vec{s}\nu}}\}_{\vec{s}}$ 
($\{\ket{\varphi_{\vec{s}\nu}}\}_{\vec{s}}$) constitutes a basis of ${\rm Span}\{\ket{\psi_{\vec{s}\nu}}\}_{\vec{s}}$ 
(${\rm Span}\{\ket{\varphi_{\vec{s}\nu}}\}_{\vec{s}}$). 
Let $S^{\psi}_\nu$ ($S^{\varphi}_\nu$) be the set of labels $\vec{s}$ of $\ket{\psi_{\vec{s}\nu}}$ ($\ket{\varphi_{\vec{s}\nu}}$) contained in such a subset. 
If $\vec{t}\notin S^{\psi}_\nu$, or equivalently, if $\ket{\psi_{\vec{t}\nu}}$ is not in the basis, it can be written 
with $d_{\vec{s}}\in\mathbb{C}$ as
\begin{equation}
\ket{\psi_{\vec{t}\nu}}=\sum_{\vec{s}\in S^{\psi}_\nu}d_{\vec{s}}\ket{\psi_{\vec{s}\nu}}.
\end{equation}
By considering the tensor-product of this equation and $\otimes\ket{0}_{E_2}$, and then multiplying it with $(\bra{\psi_{\vec{t}1}}\otimes\hat{\II}_{E_2})\hat{U}_{\nu}$ from the left and using Eq.~(\ref{eq:U_nu_def}), 
we obtain the expression of $\ket{\varphi_{\vec{t}\nu}}$ as a linear combination of $\{\ket{\varphi_{\vec{s}\nu}}\,|\,\vec{s}\in S^{\psi}_\nu\}$. 
This means that $\vec{t}\notin S^{\psi}_\nu$ implies $\vec{t}\notin S^{\varphi}_\nu$, which gives 
$\vec{t}\in S^{\varphi}_\nu\rightarrow \vec{t}\in S^{\psi}_\nu$ and $S^{\varphi}_\nu\subseteq S^{\psi}_\nu$ 
by contraposition. This results in $|S^{\varphi}_\nu|\le|S^{\psi}_\nu|$ and thus Eq.~(\ref{eq:dim_span_varphi}). 
Here, $|S|$ denotes the cardinality of set $S$.
About the dimension of vector space ${\rm Span}\{\ket{\psi_{\vec{s}\nu}}\}_{\vec{s}}$, we 
have the following lemma, which we will prove in Appendix \ref{sec:proof_Lmm_psi_dim}.
\begin{Lmm}
\label{lmm:dim_span_psi}
For any $\nu\le n-2$,
\begin{equation}
\dim {\rm Span}\{\ket{\psi_{\vec{s}\nu}}\}_{\vec{s}}<2^{n-1}.
\end{equation}
\end{Lmm}
Combining this lemma and Eq.~(\ref{eq:dim_span_varphi}) leads to
\begin{equation}
\dim {\rm Span}\{\ket{\varphi_{\vec{s}\nu}}\}_{\vec{s}}< 2^{n-1}
\label{eq:dim_varphi_2n-1}
\end{equation}
for any $\nu\le n-2$.
Based on Eq.~(\ref{eq:dim_varphi_2n-1}), our goal in the proof of Lemma \ref{lmm:HAE_lower_bound} 
can be simplified as follows. 

States $\{\ket{\varphi_{\vec{s}\nu}}\}_{\vec{s}}$ in Eq.~(\ref{eq:U_nu_def}) are shown to satisfy
Eq.~(\ref{eq:dim_varphi_2n-1}).
By substituting these states into Eqs.~(\ref{eq:rho_nu_det}), (\ref{eq:rho_nui_det}), and (\ref{eq:rho_nubi_det}), 
we obtain $\hat{\rho}_{\nu}^{\rm det}$; and by normalizing it, we obtain $\bar{\hat{\rho}}_{\nu}^{\rm det}$.
Our goal is then to show that $H(A|E)_{\bar{\hat{\rho}}_{\nu}^{\rm det}}$ for such $\bar{\hat{\rho}}_{\nu}^{\rm det}$ has a positive lower bound.

Suppose that we relax the constraints on $\{\ket{\varphi_{\vec{s}\nu}}\}_{\vec{s}}$ by imposing only Eq.~(\ref{eq:dim_varphi_2n-1}) rather than imposing Eqs.~(\ref{eq:U_nu_def}) and (\ref{eq:dim_varphi_2n-1}). 
If we can still prove under this relaxed constraint that 
the corresponding $H(A|E)_{\bar{\hat{\rho}}_{\nu}^{\rm det}}$ has a positive lower bound, then 
that will also be sufficient for our goal
\footnote{
More precisely, once we prove 
$H(A|E)_{\bar{\hat{\rho}}_{\nu}^{\rm det}}\ge C$ for $C>0$ with $\bar{\hat{\rho}}_{\nu}^{\rm det}$ under 
only Eq.~(\ref{eq:dim_varphi_2n-1}), then 
actual $H(A|E)_{\bar{\hat{\rho}}_{\nu}^{\rm det}}$ with $\bar{\hat{\rho}}_{\nu}^{\rm det}$ under both 
Eqs.~(\ref{eq:U_nu_def}) and (\ref{eq:dim_varphi_2n-1}) also satisfies $H(A|E)_{\bar{\hat{\rho}}_{\nu}^{\rm det}}\ge C$.
}. 
As this is in fact possible, hereafter we will set it as a new goal. 
To restate this new goal:
For any $d<2^{n-1}$, suppose that a set of states
\begin{align}
&V_d=\{(\ket{\varphi_{(0,\dots,0,0)\nu}},\ket{\varphi_{(0,\dots,0,1)\nu}},\dots,\ket{\varphi_{(1,\dots,1,1)\nu}})\,|\nonumber\\
&\quad\sum_{\vec{s}}\frac{\braket{\varphi_{\vec{s}\nu}|\varphi_{\vec{s}\nu}}}{2^n}
=\frac{n}{n-1},\,\dim {\rm Span}\{\ket{\varphi_{\vec{s}\nu}}\}_{\vec{s}}=d\},
\end{align}
is given and the corresponding $\hat{\rho}^{\rm det}_\nu$ is defined by substituting them into Eqs. (\ref{eq:rho_nu_det}), and (\ref{eq:rho_nui_det}), (\ref{eq:rho_nubi_det}).
The condition $\sum_{\vec{s}}
\frac{\braket{\varphi_{\vec{s}\nu}|\varphi_{\vec{s}\nu}}}{2^n}=\frac{n}{n-1}$ is meant 
to guarantee that ${\rm Tr}\hat{\rho}_{\nu}^{\rm det}=1$
\footnote{
This can be confirmed from Eqs.~(\ref{eq:rho_nu_det})-(\ref{eq:rho_nubi_det}) as
\begin{align}
{\rm Tr}\hat{\rho}_\nu^{\rm det}&=\sum_{i=1}^{n-1}\sum_{b=0,1}
\sum_{\{\vec{s}\,|\, s_i\oplus s_{i+1}=b\}}
\frac{\expect{\varphi_{\vec{s}\nu}|\varphi_{\vec{s}\nu}}}{n2^n}\\
&=\sum_{i=1}^{n-1}\frac1{n}\left(\sum_{\vec{s}}\frac{\expect{\varphi_{\vec{s}\nu}|\varphi_{\vec{s}\nu}}}{2^n}\right)\\
&=\sum_{i=1}^{n-1}\frac1{n}\frac{n}{n-1}=1.
\end{align}
}.

Then our goal now is to prove that $H(A|E)_{\hat{\rho}_{\nu}^{\rm det}}$ for such $\hat{\rho}^{\rm det}_\nu$ 
has a positive lower bound. Set $V_d$ defined above is compact.
$H(A|E)_{\hat{\rho}_{\nu}^{\rm det}}$ can be regarded a continuous map that maps $\hat{\rho}_{\nu}^{\rm det}$ 
to a real value; see, e.g., Ref. \cite{Watrous_2018}, Sec. 5.2.2. 
Thus the lower bound on $H(A|E)_{\hat{\rho}_{\nu}^{\rm det}}$ is in fact the minimum.
Therefore it suffices to show that the minimum is positive.

Suppose on the contrary that the minimum is zero.
Then it follows that $H(A|E)_{\hat{\rho}_{\nu}^{\rm det}}=\frac1{n-1}\sum^{n-1}_{i=1} 
H(A|E_2)_{\bar{\hat{\rho}}^{\rm det}_{\nu i}}=0$.
As $\bar{\hat{\rho}}^{\rm det}_{\nu i}$ are cq states, it also follows that $H(A|E_2)_{\bar{\hat{\rho}}^{\rm det}_{\nu i}}\ge0$ \cite{9781107002173}, and thus $H(A|E_2)_{\bar{\hat{\rho}}^{\rm det}_{\nu i}}=0$ for all $i$.

Further, note that the following lemma holds, which we will prove in Appendix \ref{sec:proof_Lmm_HAE_zero}.
\begin{Lmm}
\label{lmm:HAE_zero}
If $H(A|E)_{\hat{\sigma}}=0$ holds for a given cq state
\begin{equation}
\hat{\sigma}=\sum_{b=0,1}\ket{b}\bra{b}_A\otimes(\hat{\sigma}_{b})_{E},
\end{equation}
then the supports of $\hat{\sigma}_{0}$ and $\hat{\sigma}_{1}$ are disjoint.
\end{Lmm}
Then we see that, for each $i$, states $\hat{\rho}^{\rm det}_{\nu 0i}$ and 
$\hat{\rho}^{\rm det}_{\nu 1i}$ must have disjoint supports.
This means that, states $\ket{\varphi_{\vec{s}\nu}}$ and $\ket{\varphi_{\vec{s'}\nu}}$ are orthogonal to each other if 
$s_i\oplus s_{i+1}$ and $s'_i\oplus s'_{i+1}$ differ for any of $i=1,\dots,n-1$.
This implies that 
among 
$\{\ket{\varphi_{\vec{s}\nu}}\}_{\vec{s}}$, $\ket{\varphi_{\vec{s}\nu}}$ with different $(s_1\oplus s_2,s_2\oplus s_3,\cdots,
s_{n-1}\oplus s_n)$ must become basis states of vector space $\{\ket{\varphi_{\vec{s}\nu}}\}_{\vec{s}}$. 
The number of such $\vec{s}$ is $2^{n-1}$, and hence 
$\{\ket{\varphi_{\vec{s}\nu}}\}_{\vec{s}}$ includes no less than $2^{n-1}$ linearly independent elements, namely, 
\begin{align}
{\rm dimSpan}\{\ket{\varphi_{\vec{s}\nu}}\}_{\vec{s}}\ge2^{n-1}.
\end{align}
This contradicts Eq.~(\ref{eq:dim_varphi_2n-1}) and completes the proof of Lemma \ref{lmm:HAE_lower_bound}.

\section{Conclusion and Discussion}
In this paper, we have derived upper and lower bounds on the key rate for the DPS protocol with block-wise phase-randomized coherent states and revealed that these bounds coincide in the high-loss regime. When the number of pulses constituting a block is sufficiently large, the key rate becomes linearly proportional to channel transmission, which is the same scaling as the decoy BB84 and the theoretical limit of the QKD protocol. 
Therefore, our results suggest that the DPS protocol possesses both advantages 
of simple implementation and a high key generation rate. 
Given that the DPS protocol is easier to implement than other major protocols, 
our results strongly suggest that it is particularly promising and advantageous.

As a future work, it would interesting to similarly derive tight scalings of key rates for other variants of the DPS protocol. 
These variants include those without block-wise phase randomization on Alice's side, studied in Refs.~\cite{tdps1,tdps2,tdps3}, those using threshold detectors on Bob's side, 
the round-robin DPS protocol~\cite{2014Natur.509..475S}, 
the small-number-random DPS protocol~\cite{hatake2017}, and the differential quadrature phase shift protocol~\cite{dqps}. 
Moreover, it would 
be interesting to explore the mathematical technique developed in this paper not only to determine the scaling, but also 
to ascertain the values of the key rate $G$ itself in general situations.

\section*{Acknowledgments}
The authors thank Kiyoshi Tamaki, Yuki Takeuchi, and Toshihiko Sasaki for helpful comments.
A.M. is partially supported by JST, ACT-X Grant No. JPMJAX210O, Japan and 
by JSPS KAKENHI Grant Number JP24K16977. 
T.T. was supported in part by ``ICT Priority Technology Research and Development Project'' (JPMI00316) of the Ministry of Internal Affairs and Communications, Japan.

\appendix

\section{Rigorous definition of the key rate}
\label{sec:def_key_rate_exponent}
Here, we rigorously define the key rate, which was introduced in Def.~\ref{dfn:key_rate_exponent}.
We define $\ell(n,e,\eta,\varepsilon,N)$ to be the maximum length of the secret key that can be generated with block length $n$, QBER $e$, channel transmission $\eta$, security parameter $\varepsilon$, and the total number of emitted blocks $N$. 
More precisely, $\ell(n,e,\eta,\varepsilon,N)$ is the maximum number of $\varepsilon$-secure secret key bits 
that Alice and Bob can extract by using the protocol with parameters $n,N$, against any attacks by Eve in which the pulse arrival rate to Bob equals that in the absence of Eve (for the present protocol, $r$) and in which QBER=$e$. 
In general, $\ell(n,e,\eta,\varepsilon,N)$ is monotonically increasing with $\eta,\varepsilon$ and monotonically decreasing with $e$.

The asymptotic key rate corresponding to the key length $\ell$ above with the block length $N\to\infty$ is
\begin{equation}
G(n,e,\eta,\varepsilon)=\sup_N \frac{\ell(n,e,\eta,\varepsilon,N)}{N}.
\label{eq:asymptotic_key_rate}
\end{equation}
This key rate is always finite since the quantity inside sup on the right hand side is bounded. 
We also define the key rate in the limit $\varepsilon\to0$ as
\begin{equation}
G(n,e,\eta)=\lim_{\varepsilon\to0}G(n,e,\eta,\varepsilon)=\lim_{\varepsilon\to0}\sup_N \frac{\ell(n,e,\eta,\varepsilon,N)}{N}.
\label{eq:zero_epsilon_limit_key_rate}
\end{equation}
Note that this version of the key rate is independent of $\varepsilon$. 
This definition is meant to incorporate the typical situation in QKD protocols where one can achieve an arbitrarily small security parameter $\varepsilon$ by choosing the block length $N$ be sufficiently large, while keeping $\ell\ge GN$ constant.
It should be noted that $\lim_{\varepsilon\to0}$ in Eq.~(\ref{eq:zero_epsilon_limit_key_rate}) is always convergent since 
$G(n,e,\eta,\varepsilon)$ is nonnegative, and monotonically increasing with $\varepsilon$.

The Landau symbol $O(\eta^c)$ appearing in Eq.~(\ref{eq:upper}) is defined as
\begin{align}
&G(n,e,\eta)=O(\eta^{c})\nonumber\\
&\Leftrightarrow\,\exists k>0, \eta_0>0,\,\forall \eta<\eta_0,\,G(n,e,\eta)\le k\eta^{c}.
\label{eq:landau}
\end{align}
This roughly means that $G$ converges {\it faster} than $\eta^{c}$. 

On the contrary, the symbol $\Omega(\eta^c)$ appearing in Eq.~(\ref{eq:lower}) is defined as
\begin{align}
&G(n,e,\eta)=\Omega(\eta^{c})\nonumber\\
&\Leftrightarrow\,\exists k>0, \eta_0>0,\,\forall \eta<\eta_0,\,G(n,e,\eta)\ge k\eta^{c}.
\label{eq:Omega}
\end{align}
This roughly means that $G$ converges {\it slower} than $\eta^{c}$.

\section{Variable-length protocols}
\label{sec:appVL}
In our analysis in Sec.~\ref{subsec:ana}, 
we assumed that Alice and Bob abort the protocol when the measured values of detection rate $P^{\rm det}$ and error occurrence probability $P^{\rm err}$ deviate from a certain tolerance range $f^{\rm det}, f^{\rm err }$  (see 
Sec.~\ref{sec:collective_key_distillation}).
This type of protocols is useful when the quantum channel is relatively stable and $P^{\rm det}$ and $P^{\rm err}$ fluctuate moderately.
However, in practice, there are of course situations where the channel is unstable and $P^{\rm det}$ and $P^{\rm err}$ fluctuate significantly from round to round.

In order to reduce the probability of protocol abortion and optimize the average capacity in such unstable situations, some protocols have been devised in which Alice and Bob  adaptively change their parameters of key distillation (the size of the error correction code and the length of sacrificed bits in privacy amplification) depending on the observed values of $P^{\rm det}$ and $P^{\rm err}$; see, e.g., Ref. \cite{Hayashi_2012} for the case of the BB84 protocol.
These protocols are often called variable length protocols.

We note that our analysis above also carries over to variable length protocols and yields essentially the same results.
This can be done as follows.

For simplicity, set $f^{\rm err}=0$ as before.
On the other hand, as for the pulse arrival rate, choose multiple values $f^{\rm det}_1,f^{\rm det}_2,\dots,f ^{\rm det}_m$, and prove the $\varepsilon$-security respectively for $m$ protocols each of which employs one of $f^{\rm det}_i$, by using the same argument as above.
Then no matter which $f^{\rm det}_i$ Alice and Bob choose after the completion of quantum communication, $m\varepsilon$-security can be guaranteed.
This of course guarantees the $m\varepsilon$-security for the case where Alice and Bob choose the optimal $f^{\rm det}_i$ based on the measured values of $P^{\rm det}$ and $P^{\rm err}$.
That is, $m\varepsilon$-security can be achieved for variable length protocols.

Then by taking the limit $N\to\infty$, one can show that the corresponding asymptotic key rates converge to those obtained in this section.
\\

\section{Proof of Lemma \ref{lmm:dim_span_psi}}
\label{sec:proof_Lmm_psi_dim}
Here, we prove Lemma \ref{lmm:dim_span_psi}, which is copied below for convenience.
\\
{\bf Lemma~2.}~~{\it 
For any $\nu\le n-2$,}
\begin{equation}
\dim {\rm Span}\{\ket{\psi_{\vec{s}\nu}}\}_{\vec{s}}<2^{n-1}.
\end{equation}
\begin{proof}
We first define the Fourier transform of $\{\ket{\psi_{\vec{s}\nu}}\}_{\vec{s}}$ as
\begin{equation}
\ket{\tilde{\psi}_{\vec{t}\nu}}:=\frac1{\sqrt{2^n}}\sum_{\vec{s}}(-1)^{\vec{t}\cdot \vec{s}}\ket{\psi_{\vec{s}\nu}},
\label{eq:tilde_psi_def}
\end{equation}
where $\vec{t}=(t_1,\dots,t_n)\in\{0,1\}^n$.
Since the inverse Fourier transform can also be defined  similarly, we have $\dim {\rm Span}\{\ket{\psi_{\vec{s}\nu}}\}_{\vec{s}}=\dim {\rm Span}\{\ket{\tilde{\psi}_{\vec{t}\nu}}\}_{\vec{t}}$.
Thus it suffices to show that $\dim {\rm Span}\{\ket{\tilde{\psi}_{\vec{t}\nu}}\}_{\vec{t}}\ge 2^{n-1}$ holds when $\nu\ge n-1$.

By expanding Eq.~(\ref{eq:tilde_psi_def}), we have
\begin{eqnarray}
\ket{\tilde{\psi}_{\vec{t}\nu}}&=&\frac1{\sqrt{2^n}}\sum_{\vec{s}}(-1)^{\vec{t}\cdot\vec{s}}\frac1{\sqrt{\nu!}}
\nonumber\\
&&\times\sum_{\substack{m_1,\dots,m_n\ {\rm s.t.}\\ m_1+\cdots+m_n=\nu}}\frac{\nu!}{m_1!\cdots m_n!}\prod_{i=1}^{n}\left(\frac{(-1)^{s_i}\hat{a}_i^\dagger}{\sqrt{n}}\right)^{m_i}\ket{0}\nonumber\\
&=&\frac1{\sqrt{\nu!2^n n^{\nu}}}\sum_{\substack{m_1,\dots,m_n\ {\rm s.t.} \\ m_1+\cdots+m_n=\nu}}
\frac{\nu!}{m_1!\cdots m_n!}
\nonumber\\
&&\times\prod_{i=1}^{n}\left(\sum_{s_i=0,1}(-1)^{s_i(t_i+m_i)}\hat{a}_i^{\dagger m_i}\right)\ket{0}\nonumber\\
&=&\sqrt{\frac{2^n}{\nu!n^{\nu}}}\sum_{\substack{m_1,\dots,m_n\ {\rm s.t.} \\ m_1+\cdots+m_n=\nu}}\frac{\nu!}{m_1!\cdots m_n!}
\nonumber\\
&&\times\prod_{i=1}^{n}\left(1[t_i+m_i=0\ {\rm mod}\ 2]\hat{a}_i^{\dagger m_i}\right)\ket{0},
\label{eq:plemma3}
\end{eqnarray}
where ``$1[P]$'' denotes the function that equals 1 when condition $P$ holds, and 0 otherwise.
From Eq.~(\ref{eq:plemma3}), $\ket{\tilde{\psi}_{\vec{t}\nu}}$ is a linear combination of terms 
$\hat{a}_1^{\dagger m_1}a_2^{\dagger m_2}\cdots \hat{a}_n^{\dagger m_n}$.
Note that, for the coefficient of 
$\hat{a}_1^{\dagger m_1}\hat{a}_2^{\dagger m_2}\cdots \hat{a}_n^{\dagger m_n}$ in 
$\ket{\tilde{\psi}_{\vec{t}\nu}}$ to be nonzero, 
$m_i$ must all satisfy $t_i+m_i=0$ mod 2;
that is, it must hold for all $i$ that $t_i$ and $m_i$ have the same parity (even or odd).
Hence, for any $\nu\ge0$, the necessary condition of $\ket{\tilde{\psi}_{\vec{t}\nu}}\neq0$ is 
${\rm Par}[\nu]={\rm Par[\rm wt}(\vec{t})]$, and we have
\begin{align}
&|\{\vec{t}\in\{0,1\}^n|\ket{\tilde{\psi}_{\vec{t}\nu}}\neq0\}|\notag\\
\le&|\{\vec{t}\in\{0,1\}^n|{\rm Par}[\nu]={\rm Par[\rm wt}(\vec{t})]\}|=2^n/2.
\end{align}
Here, $|S|$ denotes the cardinality of set $S$.
Here, Par$[\cdot]$ and wt$(\cdot)$ denote the parity and the Hamming weight, respectively. 
This means that at least half of the vectors $\{\ket{\tilde{\psi}_{\vec{t}\nu}}\}_{\vec{t}}$ result in zero. 
If we assume (on the contrary to what we aim to prove) that $\dim {\rm Span}\{\ket{\tilde{\psi}_{\vec{t}\nu}}\}_{\vec{t}}\ge 2^{n-1}$, the half of the vectors $\vec{t}$ with different parity of $\nu$ vanish, which implies that 
the state $\ket{\tilde{\psi}_{\vec{t}\nu}}$ with $t_1=\cdots=t_{n-1}=1$ and with $t_n$ properly chosen must be nonzero. 
Such $\ket{\tilde{\psi}_{\vec{t}\nu}}$ consists of terms 
$\hat{a}_1^{\dagger m_1}\hat{a}_2^{\dagger m_2}\cdots\hat{a}_n^{\dagger m_n}$ with odd exponents 
$m_1,\dots,m_{n-1}$. This implies that $m_1,\dots,m_{n-1}\ge 1$ and thus $\nu\ge n-1$.
This completes the proof.
\end{proof}

\section{Proof of Lemma \ref{lmm:HAE_zero}}
\label{sec:proof_Lmm_HAE_zero}
In this section, we prove Lemma \ref{lmm:HAE_zero}, copied below for convenience.\\
{\bf Lemma~3.}~~{\it
If $H(A|E)_{\hat{\sigma}}=0$ holds for a given cq state}
\begin{equation}
\hat{\sigma}=\sum_{b=0,1}\ket{b}\bra{b}_A\otimes(\hat{\sigma}_{b})_{E},
\end{equation}
{\it then the supports of $\hat{\sigma}_{0}$ and $\hat{\sigma}_{1}$ are disjoint.
}\\
\begin{proof}
If $H(A|E)_{\hat{\sigma}}=0$, from the definition of the conditional entropy in Eq.~(\ref{eq:defCE}), we have
\begin{align}
H(E)_{\hat{\sigma}}=H(AE)_{\hat{\sigma}}\Longleftrightarrow H({\rm Tr}_A\hat{\sigma})=H(\hat{\sigma}),
\end{align}
which leads to
\begin{eqnarray}
H\left(\sum_{b=0,1}p_b\bar{\hat{\sigma}}_b\right)&=&\sum_{b=0,1} p_bH(\bar{\hat{\sigma}}_b)+h_2(p_0).
\end{eqnarray}
Here, we define
\begin{eqnarray}
p_b&=&{\rm Tr}(\hat{\sigma}_b),\\
\bar{\hat{\sigma}}_b&=&p_b^{-1}\hat{\sigma}_b,
\end{eqnarray}
and $h_2$ denotes the binary entropy.
Then due to Theorem 11.10 of Ref. \cite{9781107002173}, $\hat{\sigma}_0$ and $\hat{\sigma}_1$ 
have support on orthogonal subspaces.
\end{proof}


\begin{thebibliography}{99}
\bibitem{Lo2014}		
H.-K. Lo, M. Curty, and K. Tamaki, Secure quantum key distribution, Nature Photonics {\bf 8}, 595 (2014). 
\bibitem{dps1}
K. Inoue, E. Waks, and Y. Yamamoto, Differential phase shift quantum key distribution, Phys. Rev. Lett. {\bf 89}, 037902 (2002). 
\bibitem{dps2}		
K. Inoue, E. Waks, and Y. Yamamoto, Differential-phase-shift quantum key distribution using coherent light, Phys. Rev. A {\bf 68}, 022317 (2003). 
\bibitem{curty}
H. Gomez-Sousa and M. Curty, Upper bounds on the performance of differential-phase-shift quantum key distribution, 
QIC {\bf 9}, 62 (2009).
\bibitem{tsurumaru}
T. Tsurumaru, Sequential attack with intensity modulation on the differential-phase-shift quantum-key-distribution protocol, Phys. Rev. A {\bf 75}, 062319 (2007).
\bibitem{wen2009}
K. Wen, K. Tamaki, and Y. Yamamoto, Unconditional security of single-photon differential phase shift quantum key distribution, Phys. Rev. Lett. {\bf 103}, 170503 (2009).
\bibitem{Endo_2022}
H. Endo, T. Sasaki, M. Takeoka, M. Fujiwara, M. Koashi, and M. Sasaki, Line-of-sight quantum key distribution with differential phase shift keying, New Journal of Physics {\bf 24}, 025008 (2022).
\bibitem{wolf}
M. Sandfuchs, M. Haberland, V. Vilasini, and R. Wolf, Security of differential phase shift QKD from relativistic principles, arXiv:2301.11340v2 (2023).
\bibitem{tdps1}
A. Mizutani, T. Sasaki, Y. Takeuchi, K. Tamaki, and M. Koashi, Quantum key distribution with simply characterized light sources, npj Quantum Information {\bf 5}, 87 (2019).
\bibitem{tdps2}
A. Mizutani, Quantum key distribution with any two in- dependent and identically distributed states, Phys. Rev. A {\bf 102}, 022613 (2020).
\bibitem{tdps3}
A. Mizutani, Y. Takeuchi, and K. Tamaki, Finite-key security analysis of differential-phase-shift quantum key distribution, 
Phys. Rev. Res. {\bf 5}, 023132 (2023).
\bibitem{tamaki2012unconditional}
K. Tamaki, M. Koashi, and G. Kato, Unconditional security of coherent-state-based differential phase shift quantum key distribution protocol with block-wise phase randomization (2012), arXiv:1208.1995.
\bibitem{Mizutani_2018}
A. Mizutani, T. Sasaki, G. Kato, Y. Takeuchi, and K. Tamaki, Information-theoretic security proof of differential-phase-shift quantum key distribution protocol based on complementarity, Quantum Science and Technology {\bf 3}, 014003 (2017).
\bibitem{ex1}
H. Takesue, E. Diamanti, T. Honjo, C. Langrock, M. M. Fejer, K. Inoue, and Y. Yamamoto, Differential phase shift quantum key distribution experiment over 105 km fibre, New Journal of Physics {\bf 7}, 232 (2005).
\bibitem{ex2}
E. Diamanti, H. Takesue, C. Langrock, M. M. Fejer, and Y. Yamamoto, 100 km differential phase shift quantum key distribution experiment with low jitter up-conversion detectors, Opt. Express {\bf 14}, 13073 (2006).
\bibitem{ex3}
H. Takesue, S. W. Nam, Q. Zhang, R. H. Hadfield, T. Honjo, K. Tamaki, and Y. Yamamoto, Quantum key distribution over a 40-db channel loss using superconducting single-photon detectors, Nature Photonics {\bf 1}, 343 (2007).
\bibitem{tokyo}
M. Sasaki, et al., 
Field test of quantum key distribution in the tokyo qkd network, Opt. Express {\bf 19}, 10387 (2011).
\bibitem{decoy}
H.-K. Lo, X. Ma, and K. Chen, Decoy state quantum key distribution, Phys. Rev. Lett. {\bf 94}, 230504 (2005).
\bibitem{Takeoka_2014}
M. Takeoka, S. Guha, and M. M. Wilde, Fundamental rate-loss tradeoff for optical quantum key distribution, Nature Communications {\bf 5}, 5235 (2014).
\bibitem{Pirandola2017}
S. Pirandola, R. Laurenza, C. Ottaviani, and L. Banchi, Fundamental limits of repeaterless quantum communica- tions, Nature Communications {\bf 8}, 15043 (2017).
\bibitem{renner2006security}
R. Renner, Security of quantum key distribution, arXiv:quant-ph/0512258 (2006).
\bibitem{securityCOW}
T. Moroder, M. Curty, C. C. W. Lim, L. P. Thinh, H. Zbinden, and N. Gisin, 
Security of distributed-phase- reference quantum key distribution, Phys. Rev. Lett. {\bf 109}, 260501 (2012).
\bibitem{attackCOW}
J. Gonz\'alez-Payo, R. Tr\'enyi, W. Wang, and M. Curty, 
Upper security bounds for coherent-one-way quantum key distribution, Phys. Rev. Lett. {\bf 125}, 260510 (2020)
\bibitem{George_2021}
I. George, J. Lin, and N. L\"{u}tkenhaus, 
Numerical calculations of the finite key rate for general quantum key distribution protocols, Physical Review Research 
{\bf 3}, 013274 (2021).
\bibitem{Watrous_2018}
J. Watrous, The Theory of Quantum Information (Cambridge University Press, 2018).
\bibitem{9781107002173}
M. A. Nielsen and I. L. Chuang, Quantum Computation and Quantum Information (Cambridge University Press, 2011).
\bibitem{2014Natur.509..475S}
 T. Sasaki, Y. Yamamoto, and M. Koashi, Practical quantum key distribution protocol without monitoring signal
disturbance, Nature {\bf 509}, 475 (2014).
\bibitem{hatake2017}
Y. Hatakeyama, A. Mizutani, G. Kato, N. Imoto, and K. Tamaki, Differential-phase-shift quantum-key-distribution protocol 
with a small number of random delays, Phys. Rev. A {\bf 95}, 042301 (2017).
\bibitem{dqps}
S. Kawakami, T. Sasaki, and M. Koashi, Security of
the differential-quadrature-phase-shift quantum key distribution, Phys. Rev. A {\bf 94}, 022332 (2016).
\bibitem{Hayashi_2012}
M. Hayashi and T. Tsurumaru, Concise and tight security analysis of the Bennett-Brassard 1984 protocol with finite key lengths, New Journal of Physics {\bf 14}, 093014 (2012).
\end{thebibliography}
\end{document}